\documentclass[journal,twoside,web]{ieeecolor}

\usepackage{graphicx}
\usepackage{subcaption}   
\usepackage{float}        
\usepackage{balance}   
\usepackage{xcolor}
\newcommand{\rev}[1]{\textcolor{black}{#1}}
\newcommand{\rr}[1]{\textcolor{black}{#1}}
\usepackage{generic}
\usepackage{cite}
\usepackage{amsmath,amssymb,amsfonts}
\usepackage{algorithmic}
\usepackage{graphicx}
\usepackage{textcomp}
\usepackage{hyperref}
\usepackage{graphicx}
\usepackage{subcaption}
\usepackage{epsfig} 
\usepackage{amsmath}

\usepackage{cite}
\usepackage{lcsys}
\usepackage{mathtools}

\newtheorem{theorem}{Theorem}[section]

\newtheorem{proposition}[theorem]{Proposition}

\newtheorem{definition}[theorem]{Definition}
\newtheorem{remark}[theorem]{Remark}

\begin{document}

\def\BibTeX{{\rm B\kern-.05em{\sc i\kern-.025em b}\kern-.08em
    T\kern-.1667em\lower.7ex\hbox{E}\kern-.125emX}}
\markboth{\journalname, VOL. XX, NO. XX, XXXX 2017}
{Author \MakeLowercase{\textit{et al.}}: Preparation of Papers for IEEE Control Systems Letters (August 2022)}

\title{Robust Direct Data-Driven Hamiltonian for Safe Set Computation under Measurement Noise and Disturbances}

\author{
    Mohammad Bajelani\textsuperscript{a}, 
    Christopher A. Strong\textsuperscript{b},
    Claire J. Tomlin\textsuperscript{b},
    Jason J. Choi\textsuperscript{c}*, \\ and
    Klaske van Heusden\textsuperscript{a}* 
    \thanks{We acknowledge the support of the Natural Sciences and Engineering Research Council of Canada (NSERC) [RGPIN-2023-03660].}%
    \thanks{* indicates equal co-advising.}
    \thanks{\textsuperscript{a} Mohammad Bajelani and Klaske van Heusden are with the University of British Columbia, Canada. {\tt\small mohammad.bajelani, klaske.vanheusden@ubc.ca}}%
    \thanks{\textsuperscript{b} Christopher A. Strong and Claire J. Tomlin are with the University of California, Berkeley, CA 94720, USA.}%
    \thanks{\textsuperscript{c} Jason J. Choi is with the University of California, Los Angeles, CA 90095, USA.}
}

\maketitle
\thispagestyle{empty}

\begin{abstract}
    \rr{Safe set computation is a fundamental challenge in safety-critical control systems, especially in direct data-driven settings where safety analysis is performed directly from noise-affected measurements, without explicit modeling.} \rr{A recently proposed method, Data-Driven Hamiltonian (DDH)}, enables reachability analysis directly from measurements, without relying on prior knowledge of the underlying system dynamics. This paper extends the DDH framework to a robust setting that accounts for measurement noise, exogenous disturbances, and sampling-induced state-velocity estimation error. \rr{A Robust Data-Driven Hamiltonian (R-DDH) is derived from noisy measurements and shown to yield a certified lower bound on the exact Hamiltonian.} This results in a provable under-approximation of the value function and an inner approximation of the associated safe set. \rr{The gap between the data-driven and exact Hamiltonians is quantified, and it is shown to converge to zero with more data in a noise-free setting with additive disturbances.} The effectiveness of the approach is shown through two case studies: a constrained double integrator and an aircraft taxiing system with a nonlinear closed-loop controller operating under perceptual uncertainty.
\end{abstract}

\vspace{-1em}
\section{Introduction}

\rr{A robust safe set is a set of states for which there exists a feedback law that renders the set robustly positively invariant with respect to all admissible disturbances and uncertainties.} Classical computation of such sets typically requires an accurate dynamical model, which is often difficult to obtain or unnecessarily complex for control synthesis. Recently, a few \emph{direct data-driven} approaches have been proposed for reachability analysis using measurement data, enabling the construction of safe sets directly from raw measurements without explicitly identifying a model. For practical applicability of data-driven methods, robustification needs to be addressed.

Recently, a few approaches have been proposed for reachability analysis using measurement data. In~\cite{lew2020adversarial}, a random set theory framework with adversarial sampling was developed, and convergence of sampled reachable sets to the convex hull of the true reachable set was proven. A scenario optimization approach with probabilistic coverage guarantee was introduced in~\cite{dietrich2025data}, yielding a reduced sample complexity. In~\cite{alanwar2021data}, matrix zonotopes were employed to bound the reachable sets of unknown systems, while~\cite{alanwar2022enhancing} leveraged side information encoded as temporal logic constraints to reduce conservatism. In~\cite{griffioen2023data}, Gaussian process state-space models were used to compute finite-horizon forward reachable sets with exact trajectory probability measures. Active learning and disagreement-based sampling are used in~\cite{chakrabarty2020active} to estimate reachable sets from trajectory data. \rr{A neural network framework for computing over-approximations of reachable sets using reachability functions was proposed in~\cite{sun2022neureach}.} In~\cite{strong2025data}, a deterministic sampling method leveraging Lipschitz continuity and successor-state pairs was proposed to iteratively construct invariant sets for unmodeled nonlinear systems. However, these approaches typically rely on explicit (data-driven) models, focus on discrete-time systems, or offer probabilistic rather than deterministic bounds. Consequently, they bypass the unique challenge of continuous-time safety certification under both measurement noise and disturbances.

As a well-established model-based technique, Hamilton–Jacobi (HJ) reachability provides formal safety guarantees by computing a value function whose zero-superlevel set characterizes the safe set \cite{chen2018hamilton}. However, classical HJ methods rely on access to an explicit system model. \rr{The Data-Driven Hamiltonian (DDH) framework was recently proposed in~\cite{choi2025data} to eliminate the need for an explicit model by constructing a lower bound on the Hamiltonian directly from data. However, it assumes access to \emph{noise-free} input, state, and state velocity (time derivative of the state) measurements and does not account for exogenous disturbances.} To enable the practical deployment of DDH, these assumptions must be relaxed and the framework robustified. The main contributions of this paper relative to the prior work in \cite{choi2025data} are:
\begin{itemize}
    \item Robust Data-Driven Hamiltonian (R-DDH) reachability is proposed that extends the nominal framework to account for measurement noise and unknown but bounded exogenous disturbances in the dataset.
    \item To account for practical sampling effects, we quantify the uncertainty induced by finite-difference state-velocity estimation and provide guidance on selecting a suitable sampling period in the presence of measurement noise.
    \item The approximation gap between model-based and data-driven Hamiltonians is quantified to evaluate and manage conservatism. This is important because robustification can introduce excessive conservatism, rendering safety methods practically unusable. By isolating the distinct effects of noise, data sparsity, and disturbances, we demonstrate that the disturbance-induced gap completely vanishes for additive disturbances.
\end{itemize}

The remainder of the paper is organized as follows. Section~\ref{sec:problem} formulates the problem and reviews the necessary preliminaries. Section~\ref{sec:method} presents the proposed Robust Data-Driven Hamiltonian (R-DDH) reachability framework. Section~\ref{sec: Model-Based vs Data-Driven Hamiltonian} provides a rigorous gap analysis between the model-based and data-driven Hamiltonians to characterize the conservatism of the approach. Section~\ref{sec:results} demonstrates the effectiveness of the proposed approach through two numerical case studies: a constrained double integrator and an aircraft taxiing system with a nonlinear closed-loop controller under perceptual uncertainty. Section~\ref{sec:conclusion} concludes the paper and discusses future research directions.

\noindent\textbf{Notation.}
The induced \(\ell_2\)-norm is denoted by \(\|\cdot\|\), and \(|\cdot|\) denotes the elementwise absolute value of a vector. The \(\ell_2\)-ball of radius \(r\) centered at \(c\in\mathbb{R}^n\) is defined as
\[
\mathcal{B}(c,r) \coloneqq \{x\in\mathbb{R}^n \mid \|x-c\|\le r\}.
\]
The axis-aligned box with half-lengths \(l\in\mathbb{R}^n\) centered at \(c\) is denoted by
\[
\mathcal{R}(c,l) \coloneqq \{x\in\mathbb{R}^n \mid |x-c|\le l\}.
\]
If the center is the origin, we write \(\mathcal{B}(r)\) and \(\mathcal{R}(l)\). The Minkowski sum of sets \(\mathcal{A}\) and \(\mathcal{C}\) is denoted by \(\mathcal{A}\oplus\mathcal{C}\), and the set difference is denoted by \(\mathcal{A}\setminus\mathcal{C}\). In the dataset \(\mathfrak{D}\), the element \(\mathfrak{D}_j^i\) denotes the \(i\)-th sample from the \(j\)-th trajectory. For vectors \(a,b\in\mathbb{R}^n\), the elementwise product is denoted by \(a\odot b\). For \(p\in\mathbb{R}^n\), \(\operatorname{sgn}(p)\in\{-1,0,1\}^n\) denotes the elementwise sign vector, with zero entries mapped to zero.
For a differentiable value function \(V(x,t)\), the gradient is denoted by
\[
\nabla_x V(x,t) \coloneqq 
\left[
\frac{\partial V}{\partial x_1},\ldots,
\frac{\partial V}{\partial x_n}
\right]^\top,
\]
and \(\partial_t V(x,t)\) denotes the partial derivative with respect to \(t\). $x$ denotes states in the state space while $\mathbf{x}(s)$ denotes a trajectory state of the system. Note that, throughout the paper, \(x\in\mathbb{R}^n\) denotes a generic state in the state space, whereas \(\mathbf{x}(s)\in\mathbb{R}^n\) denotes the state of the system trajectory at time \(s\).

\section{Problem Formulation and Background on Hamilton--Jacobi Reachability} \label{sec:problem}
\subsection{Constrained Dynamical System}

Consider a constrained, continuous-time nonlinear system evolving over a finite horizon \( s \in [-t, 0] \):
\begin{equation} \label{eq:system}
    \dot{\mathbf{x}}(s) = f\!\big(\mathbf x(s), \mathbf u(s), \mathbf d(s)\big),
\end{equation}
where \( \dot{\mathbf{x}}(s) \in \mathbb{R}^n \), \( \mathbf{x}(s) \in \mathbb{R}^n \), \( \mathbf{u}(s) \in \mathbb{R}^m \), and \( \mathbf{d}(s) \in \mathbb{R}^p \) denote the system state velocity, state, control input, and exogenous disturbance at time \(s\), respectively, with initial condition \(\mathbf{x}(-t) = x \). We assume that the function \( f  \) is Lipschitz continuous in \( \mathbf x(s) \) and \( \mathbf d(s) \). The true state and state velocity are not directly accessible; they are corrupted by measurement noise:
\begin{equation}
\label{eq:measurement}
\begin{aligned}
\mathbf{x}_m(s) &= \mathbf{x}(s) + \mathbf{n}_x(s),\\
\mathbf{v}_m(s) &= \dot{\mathbf{x}}(s) + \mathbf{n}_v(s).
\end{aligned}
\end{equation}
where \( \mathbf x_m (s)\) and \( \mathbf v_m (s)\) denote the measured state and state velocity, and  
\(\mathbf n_x(s), \mathbf n_v(s) \in \mathbb{R}^n\) represent \rr{unknown but bounded measurement noise. It should be noted that \( n_x \) and \( n_v \) can represent a wide range of effects, such as stochastic sensor noise, calibration error, and estimation error.} The disturbance, control input, and uncertainty signals are constrained to compact sets:
\begin{equation}
\label{eq:constraints}
\begin{aligned}
\mathbf{u}(s) &\in \mathcal{U}, \quad 
\mathbf{d}(s) \in \mathcal{D},\\
\mathbf{n}_x(s) &\in \mathcal{N}_x, \quad 
\mathbf{n}_v(s) \in \mathcal{N}_v.
\end{aligned}
\end{equation}
\rr{Informally, the disturbance $\mathbf  d(s)$ may act adversarially with respect to the control $\mathbf  u(s)$; that is, while the control input seeks to achieve a specified objective, the disturbance can realize a worst-case scenario. More formally,} it is assumed that the disturbance signal $d$ is unknown but bounded and chosen in reaction to the control $u$ via a causal (nonanticipative) strategy $\xi_d \in \Xi_{[-t,0]}:\mathcal U_{[-t,0]}\!\to\!\mathcal D_{[-t,0]}$. Here, $\mathcal U,\mathcal D$ are Lebesgue-measurable sets, and $\mathcal U_{[-t,0]}$ (resp. $\mathcal D_{[-t,0]}$) denotes the set of Lebesgue-measurable maps from $[-t,0]$ to $\mathcal U$ (resp. $\mathcal D$) \cite{evans1984differential}.


We consider two forms of Lipschitz continuity, which ensure that bounded variations in the arguments of \( f \) lead to bounded variations in the state velocity. The first is the \emph{uniform} representation, defined by constants \( L^x, L^d \in \mathbb{R}_{\ge 0} \) such that, for all \( x, x' \in \mathbb{R}^n \), \( d, d' \in \mathbb{R}^p \), and \( u \in \mathcal{U} \),
\begin{equation} \label{eq:spherical_lipschitz}
    \|f(x,u,d) - f(x',u,d')\| \leq L^x \|x - x'\| + L^d \|d - d'\|.
\end{equation}
This condition implies that \( f(x',u,d') \!\in\! \mathcal{B} \big(f(x,u,d), r\big) \), where the radius is 
\(
r = L^x \|x - x'\| + L^d \|d - d'\|.
\) The second form is the \emph{component-wise} representation, which uses constant matrices \( L^{x}_{\mathrm{c}} \in \mathbb{R}^{n \times n} \) and \( L^{d}_{\mathrm{c}} \in \mathbb{R}^{n \times p} \):
\begin{equation}
\label{eq:rectangular_lipschitz}
|f(x,u,d) - f(x',u,d')| \leq L^{x}_{\mathrm{c}}\,|x - x'| + L^{d}_{\mathrm{c}}\,|d - d'|.
\end{equation}
This bounds the variation of \( f  \) as \( f(x',u,d') \!\in\! \mathcal{R}\big(f(x,u,d), l\big)\), where $l = L^{x}_{\mathrm{c}}\,|x - x'| + L^{d}_{\mathrm{c}}\,|d - d'|.$ 


   
In the data-driven setting, the function \( f \) is assumed to be unknown, but a noisy sampled dataset is available, given by
\begin{equation}
\label{eq:dataset}
\begin{aligned}
\mathfrak{D} &= \Big\{\, \{\mathfrak{D}^i_j\}_{i=1}^{N_s} \,\Big\}_{j=1}^{N_T},\\
\mathfrak{D}^i_j &= \big(x_{m,j}^i,\, u_j^i,\, v_{m,j}^i \big),
\end{aligned}
\end{equation}
where \( N_T \) is the number of trajectories and \( N_s \) is the number of samples in each trajectory. We assume uniform sampling with sampling period \( T_{\text{s}} > 0 \). The \( i \)-th sample in trajectory \( j \) is taken at time \( s = i T_{\text{s}} \), where \( x_{m,j}^i = x_{m}(i T_{\text{s}}) \), and the quantities \( u_j^i \) and \( v_{m,j}^i \) are defined analogously. For analysis, we denote the unknown true state, state velocity, and disturbance at sample \((i,j)\) by
\(
x_j^i, \,
v_j^i = f\big(x_j^i,\,u_j^i,\,d_j^i\big), \,
d_j^i,
\)
and the unknown measurement noise terms by
\(
n_{x,j}^i, \, n_{v,j}^i.
\)
Therefore, the measured state and state velocity satisfy
\begin{equation}
\label{eq:meas_rel}
\begin{aligned}
x_{m,j}^i &= x_j^i + n_{x,j}^i,\\
v_{m,j}^i &= v_j^i + n_{v,j}^i.
\end{aligned}
\end{equation}

\begin{remark}
    \rr{Unlike LTI systems, which can be represented using finite data under persistency of excitation~\cite{Ljung1999}, system identification for nonlinear systems depends strongly on the underlying dynamics and the quality of the data. In this paper, we adopt a Lipschitz continuity assumption to avoid deriving an explicit model and instead enable reachability analysis directly from data.}
\end{remark}

\subsection{Definition of Safety and Safe Sets}
\rr{In the Hamilton-Jacobi framework, safety is characterized by a user-defined failure set
$\mathcal{X}_{\mathcal{F}} \subset \mathbb{R}^{n}$, which encodes unsafe states that must be avoided (e.g., obstacle regions), and a target set
$\mathcal{X}_{\mathcal{T}} \subset \mathbb{R}^{n}$, which encodes desirable states that the system aims to reach or remain in (e.g., a goal region or a terminal invariant set).} \rr{Informally, the safe set is the set of states from which there exists an admissible control policy that keeps the system trajectory outside \( \mathcal{X}_{\mathcal{F}} \) for a finite time horizon despite all admissible disturbance realizations and, if a target set \( \mathcal{X}_{\mathcal{T}} \) is specified, steers the trajectory into \( \mathcal{X}_{\mathcal{T}} \).} This paper considers two commonly used notions of safety: the \emph{Avoid Backward Reachable Tube} (Avoid BRT) and the \emph{Reach–Avoid Backward Reachable Tube} (Reach–Avoid BRT). 

\begin{definition}[Avoid BRT \cite{choi2021robust}] \label{def: avoid BRT}
The Avoid BRT is the set of states from which the system can be controlled to avoid entering the failure set \( \mathcal{X}_{\mathcal{F}}\) over \( s \in [-t, 0] \), for all admissible disturbances:
\color{black} \begin{equation}\label{eq: Avoid BRT}
\text{Avoid}(t, \mathcal{X}_{\mathcal{F}}) =\bigg\{ 
 x \;\bigg|\;\begin{array}{l} \exists u \in \mathcal{U}_{[-t,0]} \,  \text{s.t.}\; 
\forall \xi_d \in \Xi_{[-t,0]},\\ \mathbf x(s) \notin \mathcal{X}_{\mathcal{F}},  \forall s \in [-t, 0] 
\end{array}\bigg\}.
\end{equation}
\end{definition}
\vspace{1em}

\begin{definition}[Reach–Avoid BRT \cite{choi2021robust}] \label{def: reach-avoid BRT}
The Reach–Avoid BRT is the set of states from which the system can be controlled to reach a target set \( \mathcal{X}_\mathcal{T} \subset \mathbb{R}^{n} \) at some time \( s \in [-t, 0] \), while avoiding the failure set \( \mathcal{X}_\mathcal{F} \subset \mathbb{R}^{n} \) throughout the entire time interval \([-t,s]\), for all admissible disturbances:
\color{black} \begin{equation} \label{eq: ReachAvoid BRT}
\begin{aligned}
& \text{ReachAvoid}(t, \mathcal{X}_\mathcal{F}, \mathcal{X}_\mathcal{T}) 
=  \\ &\left\{ x \,\middle|\, 
\begin{array}{l}
\exists u \in \mathcal{U}_{[-t,0]}, \exists s \in [-t, 0] \;\text{s.t.} \; \forall {\xi_d \in \Xi_{[-t,s]}}, \\ \mathbf x(s) \in \mathcal{X}_\mathcal{T} \ \& \, \mathbf x(\tau) \notin \mathcal{X}_\mathcal{F},\ \forall \tau \in [-t, s]
\end{array}
\right\}.
\end{aligned}
\end{equation}
\end{definition}

\subsection{Background on HJ Reachability}

Hamilton–Jacobi (HJ) reachability characterizes safety in dynamical systems via the value function derived from the dynamic programming principle. \rr{The value function encodes the set of states that satisfy the safety specification induced by the previously defined avoid and reach–avoid BRTs in Definitions~(\ref{def: avoid BRT}–\ref{def: reach-avoid BRT}).} In particular, the sign of the value function determines safety: non-negative values indicate states from which the safety objective can be satisfied, whereas negative values correspond to states that violate the specification.

In the following, two formulations of value functions are considered for the Avoid and the Reach–Avoid BRTs provided in \eqref{eq: Avoid BRT}–\eqref{eq: ReachAvoid BRT}. To define the value function, it is needed to define the failure set and target set by some Lipschitz continuous functions:
\begin{equation}\label{eq:set_encoding}
\mathcal{X}_{\mathcal F}=\{x \mid g(x)\le 0\},\,
\mathcal{X}_{\mathcal T}=\{x \mid l(x)\ge 0\}.
\end{equation}
Therefore, the value functions for the Avoid and Reach–Avoid BRTs \eqref{eq: Avoid BRT}–\eqref{eq: ReachAvoid BRT} can be defined as follows, respectively:
\begin{equation} \label{eq:avoid_value}
    \color{black}V_\text{A}(x,t) :=  \min_{\xi_d \in \Xi_{[-t,0]}} \max_{u \in \mathcal{U}_{[-t,0]}} \min_{s \in [-t,0]} g\big(\mathbf x(s)\big),
\end{equation}
\begin{equation} \label{eq:reachavoid_value}
\begin{aligned}
    &V_\text{RA}(x,t) \\&:= \color{black} \min_{\xi_d \in \Xi_{[-t,s]}} \max_{u \in \mathcal{U}_{[-t,s]}}  \max_{s \in [-t,0]} \min \!\left\{ l\big(\mathbf x(s)\big), \min_{\tau \in [-t,s]} g\big(\mathbf x(\tau)\big) \right\}.    
\end{aligned}
\end{equation}
The corresponding Avoid and Reach–Avoid BRTs are the zero-superlevel sets of the value functions:
\begin{equation} \label{eq:avoid_set}
\mathrm{Avoid}(t,\mathcal{X}_{\mathcal F}) = \{\, x \mid V_\text{A}(x,t) \ge 0 \,\},
\end{equation}
\begin{equation} \label{eq:reachavoid_set}
\text{ReachAvoid}(t,\mathcal{X}_{\mathcal F},\mathcal{X}_{\mathcal T}) = \{\, x \mid V_\text{RA}(x,t) \ge 0 \,\}.
\end{equation}

The value functions in \eqref{eq:avoid_value}–\eqref{eq:reachavoid_value} represent a two-player differential game: the control seeks to reach the target (for the Reach-Avoid BRT) while avoiding the failure set, and the disturbance acts adversarially. By the dynamic programming principle~\cite{bardi1997optimal,fisac2015reach}, the value function satisfies the Hamilton--Jacobi variational inequality (HJ-VI). For the Avoid BRT, the HJ-VI is:
\begin{equation} \label{eq:avoid_brt_hjvi}
    0 = \min \big\{\, g(x) - V_\text{A}(x,t),\ -\partial_t V_\text{A}(x,t) + H\big(x,\nabla_x V_\text{A}(x,t)\big) \,\big\},
\end{equation}
with terminal condition $V_\text{A}(x,0)=g(x)$. For Reach–Avoid BRT, HJ-VI is:
\begin{equation} \label{eq:reachavoid_brt_hjvi}
\begin{aligned}
0 = \min \Big\{\, g(x) - V_\text{RA}(x,t),\ \max \big\{\, l(x) - V_\text{RA}(x,t),\ \\ -\partial_t V_\text{RA}(x,t) + H\big(x,\nabla_x V_\text{RA}(x,t)\big) \big\} \Big\},
\end{aligned}
\end{equation}
with terminal condition $V_\text{RA}(x,0)=\min\{\,l(x),\,g(x)\,\}$. The Hamiltonian encodes the worst-case interaction between control and disturbance:
\begin{equation} \label{eq:hamiltonian}
    H(x,p) = \max_{u \in \mathcal{U}} \ \min_{d \in \mathcal{D}} \ p^\top f(x,u,d),
\end{equation}
where \(p := \nabla_x V_\text{A}(x,t)\) for the Avoid BRT and \(p := \nabla_x V_\text{RA}(x,t)\) for the Reach–Avoid BRT, and in both cases \(p \in \mathbb{R}^n\) denotes the costate. Note that the Hamiltonian in~\eqref{eq:hamiltonian} is the core element of HJ-VI and it is the only term in the HJ-VI that explicitly depends on the dynamics \(f\).

\subsection{Data-Driven Hamiltonian: Overview}

In \cite{choi2025data}, a data-driven framework is proposed to compute reachable sets directly from noise-free measurements, thereby eliminating the need for an explicit system model while preserving rigorous safety guarantees. In particular, the framework constructs a lower bound on the Hamiltonian in \eqref{eq:hamiltonian}. This lower bound yields an under-approximation of the true value function and, consequently, an inner approximation of the true avoid and reach-avoid BRTs in \eqref{eq: Avoid BRT}--\eqref{eq: ReachAvoid BRT}. More specifically, the DDH framework represents the dynamics through the set of all possible state velocities at a given state, referred to as the vector field bound, \(\dot{\mathbf  x}(s) \in F(\mathbf x(s))\), where \(F(x) = \{f(x,u) \mid u \in U\}\). Note that disturbances are not considered in \cite{choi2025data}. To approximate the Hamiltonian in \eqref{eq:hamiltonian}, DDH uses the information collected at sampled states to estimate the state velocity at an arbitrary query state \(x\).


Let \(\tilde{v}_j^i := f(x,u_j^i)\) denote the true velocity at the query state \(x\) corresponding to the observed control input \(u_j^i \in \mathcal U\), and let \(v_{j}^i := f(x_j^i,u_j^i)\) denote the true velocity at the sampled state \(x_j^i\) under the same input. Note that, in the nominal setting considered in \cite{choi2025data}, \(n_x=0\) and \(n_v=0\), and therefore \(v_{m,j}^i = v_j^i\). An uncertain estimate of \(\tilde{v}_j^i\) can then be constructed around \(v_j^i\) by introducing an \emph{uncertainty set} \(\mathcal{E}(x,\mathfrak{D}_j^i)\) such that
\[
\tilde{v}_j^i \in \mathcal{E}(x,\mathfrak{D}_j^i).
\]
The set \(\mathcal{E}(x,\mathfrak{D}_j^i)\) depends on the query state \(x\) and the data point \(\mathfrak{D}_j^i\), and will be defined in detail in Section~\ref{sec:method}. In particular, $\mathcal{E}(x,\mathfrak{D}_j^i)$  bounds the possible variation in the velocity between the sampled state \(x_j^i\) and the query state \(x\). Based on this construction, the data-driven Hamiltonian is defined as
\[
  \widehat{H}(x,p)
  := \max_{(i,j)\in \mathcal{I}}\;
     \min_{v\in\mathcal{E}(x,\mathfrak{D}^{i}_{j})}
     p^\top v,
\]
where \(\mathcal{I} := \{(i,j) \mid i=1,\dots,N_s,\; j=1,\dots,N_T\}\). It is shown in \cite{choi2025data} that, provided the uncertainty set \(\mathcal{E}\) is valid, the data-driven Hamiltonian satisfies
\[
\widehat{H}(x,p) \leq H(x,p).
\]
As a result, the value function obtained from \(\widehat{H}\) provides a conservative approximation of the true value function and therefore yields an inner approximation of the nominal BRT.

\begin{remark}[Direct vs. Indirect Data-Driven Settings]
    In direct data-driven settings, the system dynamics \( f  \) is unknown, and only uncertain measurements are available. An alternative, common in indirect approaches, is to identify a model from data and then evaluate the Hamiltonian~\eqref{eq:hamiltonian}. For general nonlinear systems, however, this poses a challenging system identification problem, particularly when a rigorous \rev{uncertainty quantification} is required. In the subsequent sections, we construct uncertainty sets that capture uncertain measurement, process disturbances, and sampling effects (when the state velocity is not directly measurable), thereby enabling robust reachability analysis for continuous-time systems directly from raw sampled data without relying on an explicit model.
\end{remark}

\subsection{Problem Statement}

The objective is to compute inner approximations of the robust Avoid and Reach-Avoid BRTs~\eqref{eq: Avoid BRT}-\eqref{eq: ReachAvoid BRT} from the noisy sampled dataset~\eqref{eq:dataset}, without explicitly identifying the system dynamics~\eqref{eq:system}. Given \(\mathfrak{D}\), we first seek to construct a certified lower bound on the Hamiltonian that accounts for measurement noise, exogenous disturbances, and state-velocity estimation error when the state velocity is not directly measured. This robust data-driven Hamiltonian is then used in the HJ-VIs~\eqref{eq:avoid_brt_hjvi}--\eqref{eq:reachavoid_brt_hjvi} to obtain conservative safe-set approximations.

\section{Robust Data-Driven Hamiltonian}\label{sec:method}

In this section, we extend \cite{choi2025data} to a robust setting in which measurement noise and disturbances are present. We guarantee that the resulting approximation provides a conservative estimate of the value function and therefore induces an inner approximation of the corresponding safe sets. To this end, we first present a procedure for characterizing the uncertainty set of the state velocity from a single uncertain measurement. Subsequently, for a given query state, we select the data point in the dataset that maximizes the Hamiltonian under the worst-case realization of the associated uncertainty.

\subsection{Uncertainty representation of the state velocity for a single data point using the uniform Lipschitz constant} \label{sec: uniform results}

Consider a single data point \(\mathfrak D^i_j = \big(x_{m,j}^i,\,u_j^i,\,v_{m,j}^i\big)\) in the dataset \eqref{eq:dataset}.
Given the control input \(u_j^i\), the goal is to estimate the state velocity at a query state \(x\) under an arbitrary disturbance \(d\in\mathcal D\), i.e.,
\(\tilde v_j^i := f(x,\,u_j^i,\,d)\). \textcolor{black}{Since the state velocity $\tilde v_j^i$ is influenced by an \emph{unknown} disturbance at a \emph{query} state, we characterize it as an uncertainty set using the system’s Lipschitz continuity and measured state velocity $v_{m,j}^i$.} To achieve a rigorous bound, we account for three primary sources of uncertainty:
(i) \rev{\emph{data sparsity} (the available samples $x_{m,j}^i$ are distant from the query state $x$),}
(ii) \emph{disturbance uncertainty} (variation in \(d\)),
and (iii) \emph{measurement noise} \rev{($x_{m,j}^i \neq x^i_j$)}.

\noindent
\textbf{(i), (ii) Data sparsity and disturbance.}
The dataset is finite and does not cover the entire domain; in particular, the true state \(x_j^i\) associated with the sample \(x_{m,j}^i\) may not coincide with the query state \(x\) at which the Hamiltonian is evaluated. Moreover, the realized disturbance during data collection may take any value in \(\mathcal{D}\). By the uniform Lipschitz bound \eqref{eq:spherical_lipschitz}, for \(u_j^i \in \mathcal{U}\) and \(d, d_j^i \in \mathcal{D}\), we obtain:
\begin{equation} \label{eq: uniformIII}
    \|\tilde{v}_j^i - v_j^i\|
    \;\le\; L^x \|x - x_j^i\| \;+\; L^d \|d - d_j^i\|,
\end{equation}
where \(x_j^i\) and \(v_j^i\) are the (unknown) true state and state velocity at index \((i,j)\).
If the disturbance set is an \(\ell_2\)-ball, \(\mathcal{D}=\{d:\|d\|\le \bar{d}\}\), then \(\sup_{d,\, d_j^i \in \mathcal{D}}\|d-d_j^i\| \le 2\bar{d}\), and hence
\begin{equation}\label{eq: disturbance inq}
    \|\tilde{v}_j^i - v_j^i\|
    \;\le\; L^x \|x - x_j^i\| \;+\; 2 L^d \bar{d}.
\end{equation}
\noindent
\textbf{(iii) Measurement noise.}
The true state \(x_j^i\) is not observed directly; instead we measure
\(x_{m,j}^i = x_j^i + n_{x,j}^i\) and \(v_{m,j}^i = v_j^i + n_{v,j}^i\),
with
\begin{equation*}
\begin{aligned}
n_{x,j}^i \in \mathcal{N}_x := \{n:\|n\|\le \sigma_x\}, \,
n_{v,j}^i \in \mathcal{N}_v := \{n:\|n\|\le \sigma_v\}.
\end{aligned}
\end{equation*}
Hence,
\begin{equation}\label{eq:x}
\begin{aligned}
&\|x - x_j^i\|
= \|x - (x_{m,j}^i - n_{x,j}^i)\| \\
&\le \|x - x_{m,j}^i\| + \|n_{x,j}^i\| \le \|x - x_{m,j}^i\| + \sigma_x.
\end{aligned}
\end{equation}
similarly,
\begin{equation}
\begin{aligned}\label{eq:v}
&\|\tilde{v}_j^i - v_j^i\|
= \|\tilde{v}_j^i - (v_{m,j}^i - n_{v,j}^i)\| \\
&\le \|\tilde{v}_j^i - v_{m,j}^i\| + \|n_{v,j}^i\| \le \|\tilde{v}_j^i - v_{m,j}^i\| + \sigma_v.
\end{aligned}
\end{equation}
Combining \eqref{eq:x} and \eqref{eq:v} with \eqref{eq: uniformIII} yields
\begin{equation}\label{eq:uncertain1}
    \|\tilde{v}_j^i - v_{m,j}^i\|
    \;\le\; L^x\big(\|x - x_{m,j}^i\| + \sigma_x\big) \;+\; 2L^d\bar d \;+\; \sigma_v.
\end{equation}
Consequently, {for any disturbance \(d\in\mathcal D\), the state velocity at \(x\) under input \(u_j^i\), lies within the following \(\ell_2\)-ball:}
\begin{equation}\label{eq:C-set(1)}
\boxed{\tilde v_j^i \in \mathcal{B}\!\left(
    v_{m,j}^i,\,
    L^x \big( \|x - x_{m,j}^i\| + \sigma_x \big) + 2 L^d \bar{d} + \sigma_v
\right).}
\end{equation}

\subsection{A Practical Guide to State Velocity Estimation} \label{sec: practical}

\rr{Compared to~\cite{choi2025data}, we consider a practically motivated setting in which the dataset~\eqref{eq:dataset} contains only input--state measurements and the state velocity is not directly observed. In this case, the state velocity must be reconstructed from successive state measurements. We therefore extend our analysis to explicitly account for this estimation-induced uncertainty, as described below.}

\noindent
\textbf{State velocity estimation error.}
The true state velocity \(v_j^i\) is not directly observed. Instead, we approximate it from noisy state measurements via the finite-difference estimate
\begin{equation}\label{eq:fd_est}
    \hat v_{m,j}^i \;:=\; \frac{x_{m,j}^{i+1} - x_{m,j}^{i}}{T_{\text{s}}}.
\end{equation}
Assume the inputs \((u_j^i,d_j^i)\) remain constant over the sampling interval \(T_{\text{s}}\), and the true trajectory satisfies \(\dot x(s)=f(x(s),u_j^i,d_j^i)\) with the initial condition \(x_j^i\). A second-order Taylor expansion at time \(s\) yields
\begin{align*}\label{eq:Taylor}
    &x_{j}^{i+1}
    \;=\; x_{j}^{i}
    \;+\; f(x_j^i,u_j^i,d_j^i)\,T_{\text{s}}\\&
    \;+\; \tfrac{1}{2}\,\nabla_x f(x_j^i,u_j^i,d_j^i)\,f(x_j^i,u_j^i,d_j^i)\,T_{\text{s}}^2
    \;+\; \mathcal{O}(T_{\text{s}}^3).
\end{align*}
where $\mathcal{O}(T_{\text{s}}^3)$ denotes higher-order error terms. \rev{By substituting into \eqref{eq:fd_est} and accounting for measurement noise, we obtain}
\begin{align*}
     \hat v_{m,j}^i
    \;=\; & \frac{(x_{j}^{i+1}+n_{j}^{i+1})-(x_{j}^{i}+n_{j}^{i})}{T_{\text{s}}}
   \\  \;=\;&  f(x_j^i,u_j^i,d_j^i)
    \;+\; \frac{1}{2}\,\nabla_x f(x_j^i,u_j^i,d_j^i)\,f(x_j^i,u_j^i,d_j^i)\,T_{\text{s}}
    \\\;+\; & \frac{n_{j}^{i+1}-n_{j}^{i}}{T_{\text{s}}}
    \;+\; \mathcal{O}(T_{\text{s}}^2).
\end{align*}
Hence, the total estimation error is bounded by
\begin{align}
&\big\|\hat v_{m,j}^i - v_{j}^i\big\|
\le \tfrac{1}{2}\big\|\nabla_x f(x_j^i,u_j^i,d_j^i)\big\|\,\big\|f(x_j^i,u_j^i,d_j^i)\big\|\,T_{\text{s}} \\&
    + \frac{\|n_{j}^{i+1}-n_{j}^{i}\|}{T_{\text{s}}}
    + \mathcal{O}(T_{\text{s}}^2). \nonumber
\end{align}
\rev{Assuming that $\|f(x,u,d)\|\le M$ and, from \eqref{eq:spherical_lipschitz}, $\|\nabla_x f(x,u,d)\|\le L^x$, it follows that}
\begin{equation}\label{eq:taylor}
    \big\|\hat v_{m,j}^i - v_{j}^i\big\|
    \;\le\; \frac{1}{2}L^x M\,T_{\text{s}} \;+\; \frac{2\sigma_x}{T_{\text{s}}} \;+\; \mathcal{O}(T_{\text{s}}^2).
\end{equation}
\rev{Combining \eqref{eq:uncertain1} and \eqref{eq:taylor} via the triangle inequality and assuming $\mathcal{O}(T_{\text{s}}^2)\approx0$, it results in}
\begin{equation}\label{eq:setcombined}
\begin{aligned}
    &\big\|\tilde v_j^i - \hat v_{m,j}^i\big\|
    \le L^x\big(\|x - x_{m,j}^i\|+\sigma_x\big) \;+\; 2L^d \bar d \; \\& +\; \frac{1}{2}L^x M\,T_{\text{s}} \;+\; \frac{2\sigma_x}{T_{\text{s}}}.
\end{aligned}
\end{equation}

Consequently, the state velocity $\tilde v_j^i$ at \(x\) under input \(u_j^i\) and disturbance \(d_j^i\) lies in an \(\ell_2\)-ball:
\begin{equation}\label{eq:C-set2}
\boxed{%
\begin{aligned}
\mathcal{B}\Big(
    \hat v_{m,j}^i,\;
    &L^x \big(\|x - x_{m,j}^i\| + \sigma_x\big) + 2L^d \bar d \\
    &\quad + \frac{1}{2} L^x M T_{\text{s}} + \frac{2\sigma_x}{T_{\text{s}}}
\Big),
\end{aligned}}
\end{equation}


\begin{remark}
Comparing \eqref{eq:C-set(1)} and \eqref{eq:C-set2} shows that, when \(v_{m,j}^i\) is replaced by the finite-difference estimate \(\hat v_{m,j}^i\), the velocity-uncertainty bound can be defined as follows,
\[
\sigma_v := \frac{1}{2}L^x M T_{\mathrm{s}} + \frac{2\sigma_x}{T_{\mathrm{s}}}.
\]
Therefore, in the practical setting where the state velocity is estimated from noisy state measurements, it is sufficient to choose
\begin{equation}
\label{eq:sigma_v_condition}
\sigma_v \geq \frac{1}{2} L^x M T_{\mathrm{s}} + \frac{2\sigma_x}{T_{\mathrm{s}}}.
\end{equation}
This ensures that the finite-difference velocity estimate is contained in the prescribed uncertainty set. Hence, the robust DDH remains a valid lower bound on the exact Hamiltonian. This condition can be satisfied by overestimating the Lipschitz constant, the bound on \(f\), and the state-measurement noise.
\end{remark}

If the sampling period \(T_{\text{s}}\) can \rev{be chosen by designer} (e.g., via down sampling), the optimal value for sampling period $T_{\text{s}}^\star$ becomes:
\begin{equation}\label{eq:Ts_opt}
    T_{\text{s}}^\star \;=\; \arg\min_{\,T_{\min}<T_{\text{s}}\le T_{\max}} \left\{ \tfrac{1}{2}L^x M\,T_{\text{s}} \;+\; \frac{2\sigma_x}{T_{\text{s}}} \right\},
\end{equation}
where $T_{\min}$ and $T_{\max}$ denote the minimum and maximum user-defined sampling period. The minimizer is as follows,
\begin{equation}
    T_{\text{s}}^\star = \max\big(T_{\min},\min(2\sqrt{\sigma_x/(L^x M)},\,T_{\max})\big).
\end{equation}
If $T_{\min} \leq2\sqrt{\sigma_x/(L^x M)} \leq T_{\max}$, substituting \(T_{\text{s}}^\star\) into \eqref{eq:setcombined} yields
\begin{equation}
\begin{aligned}
\big\|\tilde v_j^i - \hat v_{m,j}^i\big\|
\le L^x\big(\|x - x_{m,j}^i\|&+\sigma_x\big) 
+ 2L^d \bar d \\&+ 2\sqrt{L^x M\,\sigma_x} \ .
\end{aligned}
\end{equation}

It should be noted that the \emph{optimal} sampling period only jointly minimizes the effects of measurement noise and state velocity estimation errors. Since changing the sampling period \(T_{\mathrm{s}}\) changes the subsequent measured sample \(x_{m,j}^{i+1}\), it also affects the estimated velocity \(v_{m,j}^{i}\). Therefore, a more general approach would be to optimize the sampling period with respect to the Hamiltonian \(H(x,p)\). 
\begin{remark}
As the sampling period \(T_{\text{s}}\) increases, higher-order truncation terms in the Taylor expansion can no longer be neglected. 
Accurate modeling in such settings requires incorporating \rev{higher-order terms}.
\end{remark}
\begin{remark}
The optimal sampling period in~\eqref{eq:Ts_opt} is used only for estimating the state velocity. In practice, the user may sample the system at a much finer rate to collect a larger dataset, while still using \(T_s^\star\) for the velocity estimation. For instance, one may record states every \(\frac{T_s}{10}\) and use \((x_m^{i},\, x_m^{i+10})\) for the velocity estimate. However, for simplicity, the formulation here is presented using two consecutive measurements separated by the optimal sampling period \(T_s^\star\).
\end{remark}

\subsection{Uncertainty representation of the state velocity for a single data point using the component-wise Lipschitz constant} \label{sec: component-wise results}

\textcolor{black}{Since the component-wise Lipschitz continuity changes only the geometry of the bound, the derivation parallels the uniform case; hence, the intermediate steps are omitted for brevity.}
Specializing the proof of Section~\ref{sec: uniform results} to the rectangular bound~\eqref{eq:rectangular_lipschitz}, and adopting
$\mathcal{D}=\{ d\in\mathbb{R}^p:\, |d|\le \bar d \}$ with
\begin{equation*}
\begin{aligned}
n_{x,j}^i \in \mathcal{N}_x := \{n:|n|\le \sigma_x\}, \,
n_{v,j}^i \in \mathcal{N}_v := \{n:|n|\le \sigma_v\}.
\end{aligned}
\end{equation*}
We obtain
\begin{align}\label{eq:square-final}
    \big| \tilde{v}_j^i - v_{m,j}^i \big|
    \le L^{x}_\mathrm{c} \big( |x - x_{m,j}^i| + \sigma_x \big) + 2 L^{d}_\mathrm{c} \bar{d} + \sigma_v,
\end{align}
and the uncertainty set is the axis-aligned hyper rectangle,
\begin{equation}\label{eq:R-set(1)}
\boxed{\mathcal{R}\!\left(
    v_{m,j}^i,\,
    L^{x}_\mathrm{c} \big( |x - x_{m,j}^i| + \sigma_x \big) + 2 L^{d}_\mathrm{c} \bar{d} + \sigma_v
\right).}
\end{equation}
Similarly, when the state velocities are not measured, as explained in Section~\ref{sec: practical}, and under the assumption of a component-wise representation, we have:
\begin{equation}\label{eq:R-set2}
\boxed{%
\begin{aligned}
\mathcal{R}\Big(
    \hat v_{m,j}^i,\;
    &L^{x}_{c}\big(|x - x_{m,j}^i| + \sigma_x\big) + 2L^{d}_{c} \bar d \\
    &\quad + \frac{1}{2} L^x_c M T_{\text{s}} + \frac{2\sigma_x}{T_{\text{s}}}
\Big).
\end{aligned}}
\end{equation}



\begin{figure*}
    \centering
    \includegraphics[width=1\linewidth, trim={0cm 4cm 0cm 3cm}, clip]{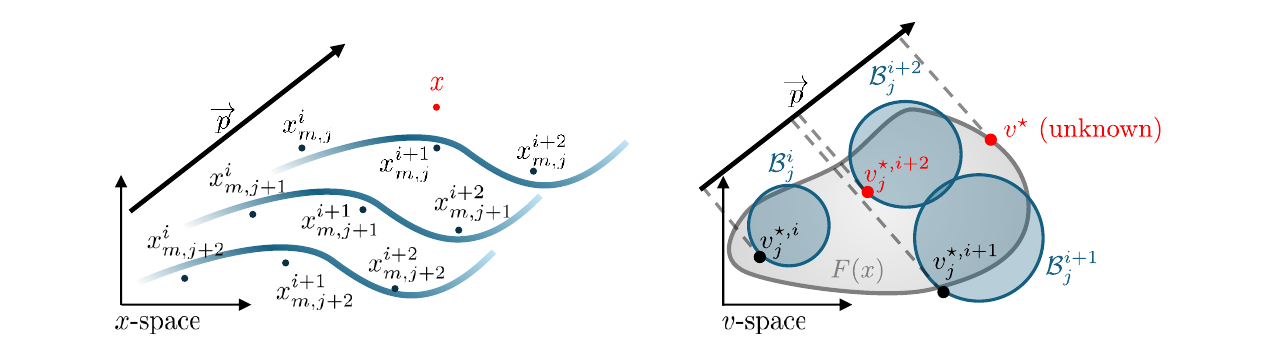}
    \caption{Visualization of the R-DDH framework. \textbf{(Left)} Three trajectories, indexed by \(j\), \(j+1\), and \(j+2\), are shown together with their sampled noisy states at indices \(i\), \(i+1\), and \(i+2\). The blue curves represent the true continuous-time trajectories, while the black markers (e.g. \(x_{m,j}^i\)) denote noisy state measurements, which may not lie on the corresponding trajectories because of measurement noise. The red point \(x\) denotes an arbitrary query state, and the vector \(\vec{p}\) indicates the costate direction. \textbf{(Right)} The gray region \(F(x)\) represents the set of all admissible state velocities at the query state \(x\). The blue circles denote the uncertainty sets associated with individual data points in the $j$-th trajectory, whose centers and radii are computed according to \eqref{eq:C-set2}, and are denoted here by \(\mathcal B_j^i\), \(\mathcal B_j^{i+1}\), and \(\mathcal B_j^{i+2}\). From the model-based perspective, the optimal state velocity is denoted by \(v^\star\), which is unknown. For each data point, the corresponding worst-case estimate is denoted by \(v_j^{\star,i}\), \(v_j^{\star,i+1}\), and \(v_j^{\star,i+2}\), defined as the point in the associated uncertainty set lying in the direction opposite to \(\vec{p}\). Among all such candidates, R-DDH selects the one that yields the largest value of \(\hat{H}(x,p)\). In this illustration, the selected point is \(v_j^{\star,i+2}\).}
    \label{fig:placeholder}
\end{figure*}
\subsection{Robust Data-Driven Hamiltonian}

\rev{Under the worst-case uncertainty, R\text{-}DDH selects the sample from the dataset \eqref{eq:dataset} that maximizes the Hamiltonian at the query state \(x\).} \rr{ This guarantees $\widehat{H}(x,p)\leq H(x,p)$} with respect to measurement error, exogenous disturbances, and state velocity estimation error, see Proposition \ref{proposition: H and H hat}. \rev{Accordingly, the set of all admissible state velocities at a given state \(x\), and a control input \(u\) under all feasible disturbances $d$ becomes:
\begin{equation}
    \dot{\mathbf  x}(s) \in F(\mathbf x(s),\mathbf  u(s)),
\end{equation}
where,
\begin{equation}
    F(x,u) := \{\, f(x,u,d) \mid d \in \mathcal{D} \,\}.
\end{equation}
Thus, the Hamiltonian in \eqref{eq:hamiltonian} can be written as
\begin{equation}\label{eq:Hamiltonian-VFB}
    H(x,p)
    \;=\;
    \max_{u \in \mathcal{U}} \; \min_{v \in F(x,u)} \; p^\top v.
\end{equation}}

\rev{Although \eqref{eq:Hamiltonian-VFB} removes explicit dependence on the model \eqref{eq:system}, it cannot be evaluated directly when the dynamics are unknown. Therefore, we estimate (or measure) the state velocity from the sampled dataset and construct a certified under-approximation of the exact Hamiltonian.} Considering the measured velocity\footnote{If the state velocity is not directly measurable, replace \(v_{m,j}^i\) with its estimate \(\hat v_{m,j}^i\) and use the uncertainty sets in \eqref{eq:C-set2} or \eqref{eq:R-set2}.}  \(v^{i}_{m,j}\), the uncertainty sets \eqref{eq:C-set(1)} or \eqref{eq:R-set(1)} provide \rev{a set of possible} \(\tilde{v}^i_j\), given state \(x\), control \(u_j^i\), and arbitrary disturbance \(d\). Formally,
\begin{equation}
    \tilde v^{i}_j \in \mathcal{E}\big(x,\mathfrak{D}^{i}_{j}\big),
\end{equation}
where \(\mathcal{E}\big(x,\mathfrak{D}^{i}_{j}\big)\) denotes the uncertainty set, either \eqref{eq:C-set(1)} or \eqref{eq:R-set(1)}, associated with the data point \(\mathfrak{D}^{i}_{j}\).

The robust data-driven Hamiltonian is then defined by minimizing the linear objective \(p^\top v\) over each data point’s uncertainty set in the dataset \eqref{eq:dataset}, and then maximizing over all data points:
\begin{equation}\label{eq:DDH}
  \widehat{H}(x,p)
  := \max_{(i,j)\in \mathcal{I}}\;
     \min_{v\in\mathcal{E}(x,\mathfrak{D}^{i}_{j})}
     p^\top v.
\end{equation}
The inner minimization captures the worst-case realization of the uncertain state velocity at state \(x\) under input \(u_j^i\) and arbitrary disturbance \(d\), while the outer maximization identifies the data point that yields the largest Hamiltonian estimate despite this worst-case uncertainty. Fig. 1 provides a visual illustration of the R-DDH framework for a small dataset.

The inner minimization in \eqref{eq:DDH} admits closed-form solutions for the hyper-sphere and hyper-rectangle uncertainty sets defined in \eqref{eq:C-set(1)} and \eqref{eq:R-set(1)}, respectively \cite{choi2025data}. For the hyper-sphere case,
\begin{equation}
\label{eq:ball-value}
\min_{v\in \mathcal B(v^{i}_{m,j},r^{i}_{j})}
     p^\top v = p^\top (v^{i}_{m,j} - \frac{p}{||p||}r^{i}_{j}(x)),
\end{equation}
where radius $r^i_j(x)$ is given in \eqref{eq:C-set(1)}, and for the hyper-rectangle case,
\begin{equation}
\label{eq:rect-value}
\min_{v\in \mathcal R(v^{i}_{m,j},l^{i}_{j})}
     p^\top v = p^\top (v^{i}_{m,j} - \operatorname{sgn}(p)\odot l^{i}_{j}(x)) ,
\end{equation}
where the half-length $l^i_j(x)$ is given in \eqref{eq:R-set(1)}. 


\begin{proposition}[Lower bound on the exact Hamiltonian] \label{proposition: H and H hat}
Let $\tilde{v}^{i}_j \in \mathcal{E}(x,\mathfrak{D}^{i}_{j})$. Then
\begin{equation*}
    \widehat{H}(x,p) \;\le\; H(x,p).
    \label{eq:ddh_lower_bound}
\end{equation*}
\end{proposition}

\begin{proof}
The exact Hamiltonian is defined as $H(x,p)
=
\max_{u\in\mathcal{U}}
\min_{v\in F(x,u)}
p^\top v$, which corresponds to choosing the optimal control input against the worst-case disturbance. In contrast, the data-driven surrogate $\widehat{H}(x,p)
=
\max_{(i,j)\in\mathcal{I}}
\min_{v\in\mathcal{E}(x,\mathfrak{D}^{i}_{j})}
p^\top v$
optimizes the same inner product, but only through data-induced uncertainty sets.

For each data point \(\mathfrak{D}^{i}_{j}\), the uncertainty set \(\mathcal{E}(x,\mathfrak{D}^{i}_{j})\) is constructed so that it contains all possible state velocities at the query state \(x\) under the sampled input \(u_j^i\) and all admissible disturbances. Therefore,
$
F(x,u_j^i) \subseteq \mathcal{E}(x,\mathfrak{D}^{i}_{j})
$ for $d \in \mathcal D$.
It follows that
\[
\min_{v\in\mathcal{E}(x,\mathfrak{D}^{i}_{j})} p^\top v
\le
\min_{v\in F(x,u_j^i)} p^\top v .
\]
Taking the maximum over all data points gives
\[
\max_{(i,j)\in\mathcal{I}}
\min_{v\in\mathcal{E}(x,\mathfrak{D}^{i}_{j})} p^\top v
\le
\max_{(i,j)\in\mathcal{I}}
\min_{v\in F(x,u_j^i)} p^\top v .
\]
Since each sampled input satisfies \(u_j^i\in\mathcal{U}\), the sampled inputs form a subset of the admissible input set. Therefore,
\[
\max_{(i,j)\in\mathcal{I}}
\min_{v\in F(x,u_j^i)} p^\top v
\le
\max_{u\in\mathcal{U}}
\min_{v\in F(x,u)} p^\top v.
\]
Hence,
\[
\widehat{H}(x,p)\le H(x,p).
\]
\end{proof}

\begin{proposition} \label{Conservative value functions and safe sets}
Let \(V(\cdot)\) and \(\widehat V(\cdot)\) denote the viscosity solutions of the HJ-VI \eqref{eq:avoid_brt_hjvi} with the exact Hamiltonian \(H\) and the robust data-driven Hamiltonian \(\widehat H\), respectively, under identical terminal conditions. Then
\[
\widehat V(x,t)\le V(x,t)\quad \forall (x,t),
\]
and consequently the safe sets satisfy
\[
\widehat S(t):=\{x:\widehat V(x,t)\ge 0\}\ \subseteq\ S(t):=\{x:V(x,t)\ge 0\}.
\]
\end{proposition}
\vspace{0.5em}
\begin{proof}
The inequality $\widehat V(x,t)\le V(x,t)$ for all $(x,t)$ follows directly from Theorem~1 in \cite{choi2025data}, and the inclusion $\widehat S(t) \subseteq S(t)$ follows directly from Theorem~2 in \cite{choi2025data}.
\end{proof}

\section{Model-Based vs Data-Driven Hamiltonian} \label{sec: Model-Based vs Data-Driven Hamiltonian}

In this section, \rev{we treat the model-based HJ reachability solution as the \emph{ideal benchmark}, since it assumes exact knowledge of the system dynamics.} By contrast, the proposed direct data-driven method bypasses modeling and operates on raw measurements, so its \rev{conservatism} depends on dataset size and sample quality. To assess the degree of conservatism induced by the proposed robustification, we quantify the Hamiltonian approximation gap,
\begin{equation} \label{eq: gap}
    \Delta H(x,p) := H(x,p) - \widehat H(x,p).
\end{equation}

For a given query state \(x\), \rev{using \eqref{eq:R-set(1)}} and \eqref{eq:rect-value}, we obtain:
\begin{align} \nonumber
& H(x,p) - \widehat{H}(x,p) 
= \max_{u \in \mathcal{U}} \min_{v \in F(x,u)} p^\top v 
- \max_{(i,j)\in \mathcal{I}}\, \min_{v \in \mathcal{R}(\cdot)} p^\top v \\
&= \max_{u \in \mathcal{U}} \min_{v \in F(x,u)} p^\top v 
- \max_{(i,j)\in \mathcal{I}}\, p^\top \left( v_{m,j}^i - {l^\star}^i_j(x) \right),
\label{eq:DDH-rect}
\end{align}
where
\begin{align} \label{eq:l}
{l^\star}^i_j(x) = \operatorname{sgn}(p) \odot  \big( L^{x}_{\mathrm{c}} \big|x - x_{m,j}^i\big| 
+ L^{x}_{\mathrm{c}} \sigma_x 
+ 2 L^{d}_{\mathrm{c}} \bar{d} 
+ \sigma_v \big) .
\end{align}
\rev{The terms \(L^{x}_{\mathrm{c}}\sigma_x\) and \(\sigma_v\) represent bias components that remain non-vanishing even as the number of data samples increases.} This persistence follows from the fact that the sets $\mathcal{N}_x$ and $\mathcal{N}_v$ impose deterministic bounds and capture a variety of perceptual uncertainties, such as measurement noise and sensor bias. The term $L^{x}_{\mathrm{c}}\lvert x - x_{m,j}^i\rvert$ captures approximation error due to data sparsity and diminishes as the dataset becomes denser; in particular, for any query state $x$, the presence of a nearby $x_{m,j}^i$ in the dataset reduces this term.

\rev{Crucially, although} the disturbance-related term $2 L^{d}_{\mathrm{c}} \bar{d}$ appears irreducible and would act as a bias even as more data are collected, Theorem~\ref{thm:ddh} shows that it can be reduced under appropriate assumptions on the system and data.

\begin{theorem}\label{thm:ddh}
Consider the system $\dot x = f(x,u) + d$ with disturbance set $\mathcal{D} := \{ d \in \mathbb{R}^n : |d| \le \bar d \}$, and noise sets $\mathcal{N}_x := \{n:|n|\le \sigma_x\},\,\mathcal{N}_v := \{n:|n|\le \sigma_v\}$. For any query state $x$ and costate $p \in \mathcal{B}(1) \subseteq \mathbb R^n$, consider
\begin{equation}\label{eq:HJB for linear d}
    H(x,p)
    \;=\;
    \max_{u \in \mathcal{U}}\;
    \min_{d \in \mathcal{D}}\;
    p^\top \big(f(x,u)+d\big),
\end{equation}
and let \eqref{eq:HJB for linear d} admit an optimizer $u^\star(x,p)$, with the minimizing disturbance
\[
d^\star(x,p) = -\bar d \odot \operatorname{sgn}(p).
\]
Also consider 
\begin{equation}\label{eq:HJB for linear d_proof}
    \widehat{H}(x,p) = \max_{(i,j)\in \mathcal{I}}\, p^\top ( v_{m,j}^i - {l^\star})
\end{equation}
and let the index \((i^\star,j^\star)\) be the optimal index in the dataset. Assume further that:  
(i) for every query state \(x\) and for every arbitrarily small \(\varepsilon > 0\), there exists a sample \(x_{m,j^\star}^{i^\star}\) in the dataset such that \(\| x - x_{m,j^\star}^{i^\star} \| \le \varepsilon\), and  
(ii) for every pair $(x,p)$, the same data point contains the maximizing control $u^\star(x,p)$ and the disturbance $-d^\star(x,p)$. Then the gap between the model-based solution in \eqref{eq:DDH-rect} and the data-driven solution satisfies
\begin{equation}
    H(x,p)-\widehat{H}(x,p)
    =
    |p|^\top (L^{x}_{\mathrm{c}} \sigma_x + \sigma_v).
\end{equation}
\end{theorem}
\vspace{1em}
\begin{proof}
Assumption (i) implies that, for any query state \(x\), there exists a sample \(x_{m,j^\star}^{i^\star}\) such that \(x_{m,j^\star}^{i^\star} \approx x\). Hence, the uncertainty set in \eqref{eq:l} reduces to
\[
{l^\star} = \operatorname{sgn}(p) \odot
\big(L^{x}_{\mathrm{c}} \sigma_x + \sigma_v + 2\,\bar d\big).
\]

Leveraging (ii), the data point $(i^\star,j^\star)$ with the optimal control input $u^\star$ and the corresponding disturbance $-d^\star$ is present in the dataset, where
\[
-d^\star = \bar d \odot \operatorname{sgn}(p).
\]
Therefore, from \eqref{eq:HJB for linear d_proof}, we obtain
\[
\begin{aligned}
\widehat H(x,p) 
&= \max_{(i,j)\in \mathcal{I}} \min_{v \in \mathcal R(\cdot)} p^\top v \\[4pt]
&= p^\top\big(f(x,u^\star) -d^\star - l^\star(x)\big) \\[4pt]
&= p^\top\big(f(x,u^\star) +\operatorname{sgn}(p)\odot\bar{d} 
\\ &- \operatorname{sgn}(p) \odot
(L^{x}_{\mathrm{c}} \sigma_x + \sigma_v + 2\,\bar d)\big) \\[4pt]
&= p^\top f(x,u^\star)
- |p|^\top \bar d
- |p|^\top (L^{x}_{\mathrm{c}} \sigma_x + \sigma_v).
\end{aligned}
\]
On the other hand, the model-based Hamiltonian satisfies
\[
\begin{aligned}
H(x,p) 
&= \max_{u \in \mathcal U}\min_{d \in \mathcal D} p^\top\big(f(x,u)+d\big) \\[4pt]
&= p^\top\big(f(x,u^\star)-\operatorname{sgn}(p) \odot \bar d\big) \\[4pt]
&= p^\top f(x,u^\star)-|p|^\top \bar d .
\end{aligned}
\]
Hence,
\begin{equation*}
    H(x,p)-\widehat{H}(x,p)=
|p|^\top (L^{x}_{\mathrm{c}} \sigma_x + \sigma_v).
\end{equation*}
\end{proof}

Informally, when the disturbance enters additively as $\dot{x}=f(x,u)+d$, the data-driven Hamiltonian approximation becomes tighter and the disturbance-related term $2L^{d}_{c}\bar d$ vanishes in the DDH formulation. Note that in the noise-free setting where $\sigma_x,\sigma_v=0$, \rev{even if all terms in \eqref{eq:l} vanish, exact recovery of the Hamiltonian, i.e., $H(x,p) = \widehat{H}(x,p)$, requires that at least one measured velocity $v_{m,j}^i$ coincide with the optimal control input and the disturbance $-d^\star(x,p)$.  
In other words, the optimal control input, which maximizes the Hamiltonian for a given query state $x$ at any time within the prescribed horizon, must be represented in the dataset.} In that case the data-driven value function satisfies $\widehat V(x,t)=V(x,t)$ and the computed safe set satisfies $\widehat{\mathcal{S}}(t)=\mathcal{S}(t)$. See Example~\ref{Number of Samples} for a numerical demonstration of this convergence behavior in the case $\bar d \neq 0$, where the DDH formulation nonetheless computes the exact robust safe set from noise-free data. 

\textcolor{black}{\begin{remark}
    The term $-d^\star(x,p)$ corresponds to a \emph{cooperative} disturbance realization, which may arise in stochastic settings but is generally absent when disturbances act adversarially during data collection.  
    In the latter case, when the disturbance is known to be adversarial, the uncertainty set in~\eqref{eq: disturbance inq} may be tightened from $2L^{d}\bar d$ to $L^{d}\bar d$ to avoid unnecessary conservatism.
\end{remark}}

\begin{figure*}[!t]
    \centering
    \begin{subfigure}[t]{0.32\textwidth}
        \centering
        \includegraphics[width=1\linewidth, trim={1cm 0 3cm 0}, clip]{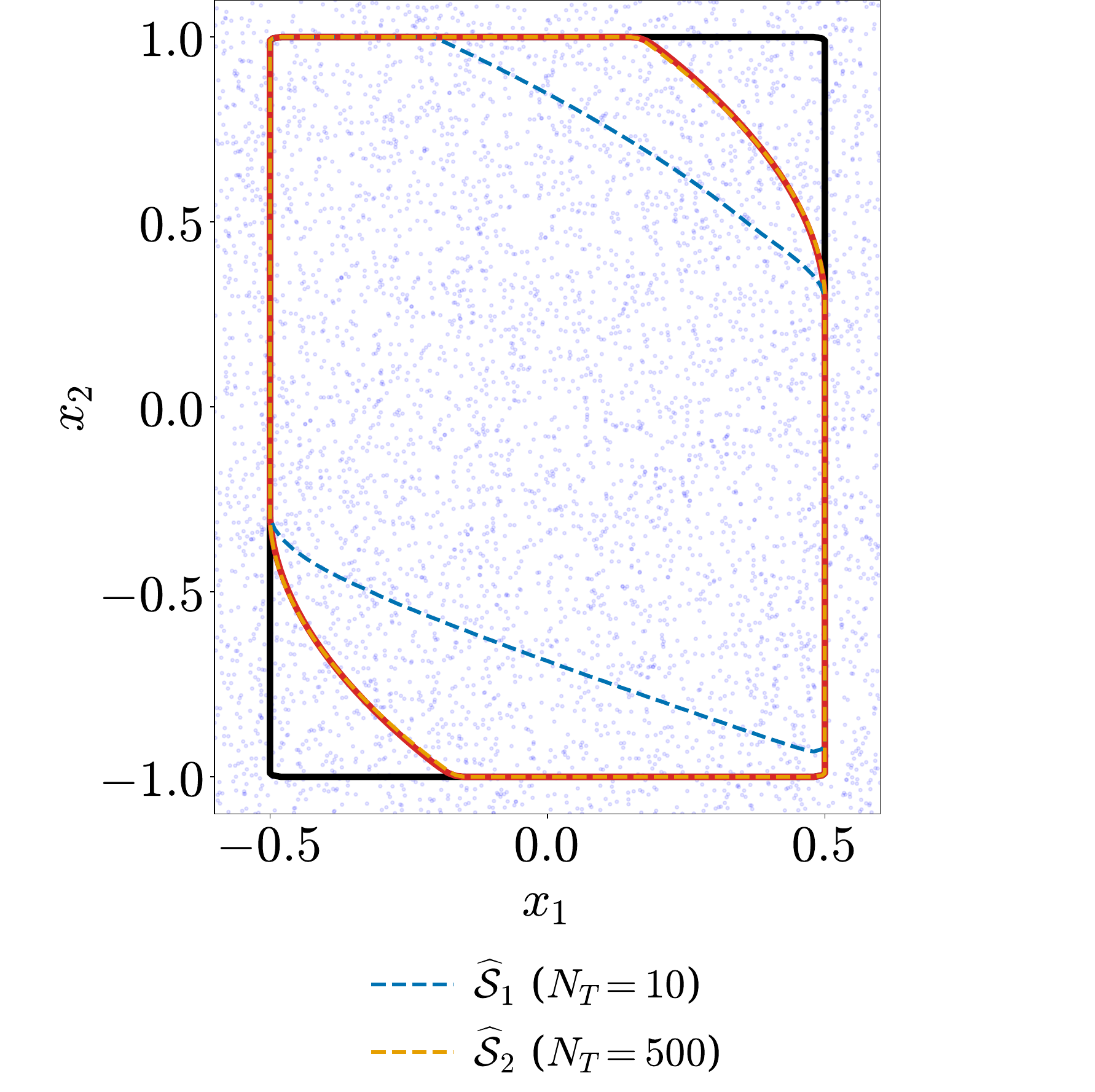}
        \caption{Effect of the number of samples \(N_T\). (Section~\ref{Number of Samples})}
        \label{fig:di_N}
    \end{subfigure}
    \hfill
    \begin{subfigure}[t]{0.32\textwidth}
        \centering
        \includegraphics[width=1\linewidth, trim={1cm 0 3cm 0}, clip]{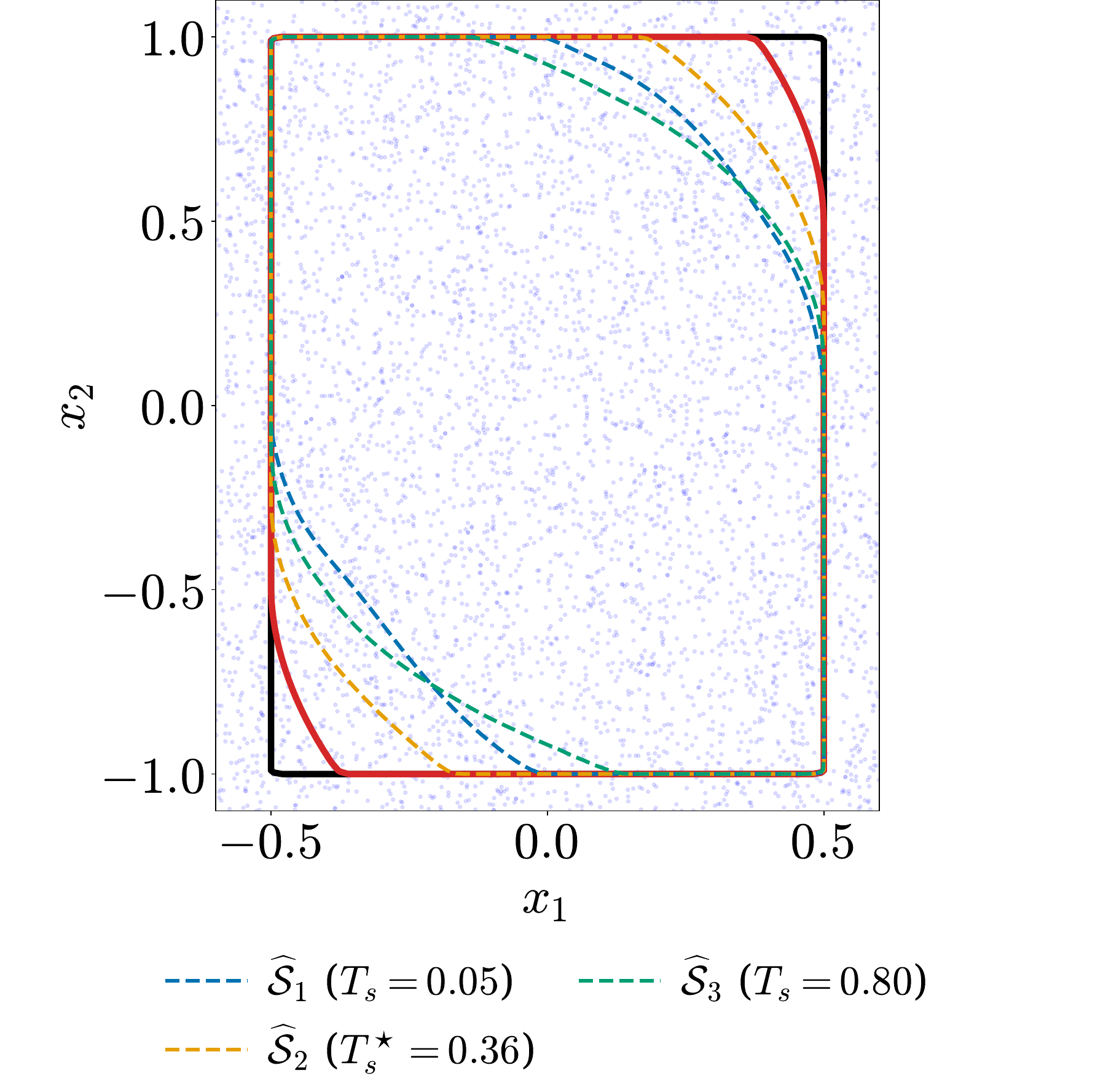}
        \caption{Effect of different sampling periods \(T_{\text{s}}\) (Section~\ref{sec:Sampling Period}).}
        \label{fig:di_Ts}
    \end{subfigure}
    \begin{subfigure}[t]{0.32\textwidth}
        \centering
        \includegraphics[width=1\linewidth, trim={1cm 0 3cm 0}, clip]{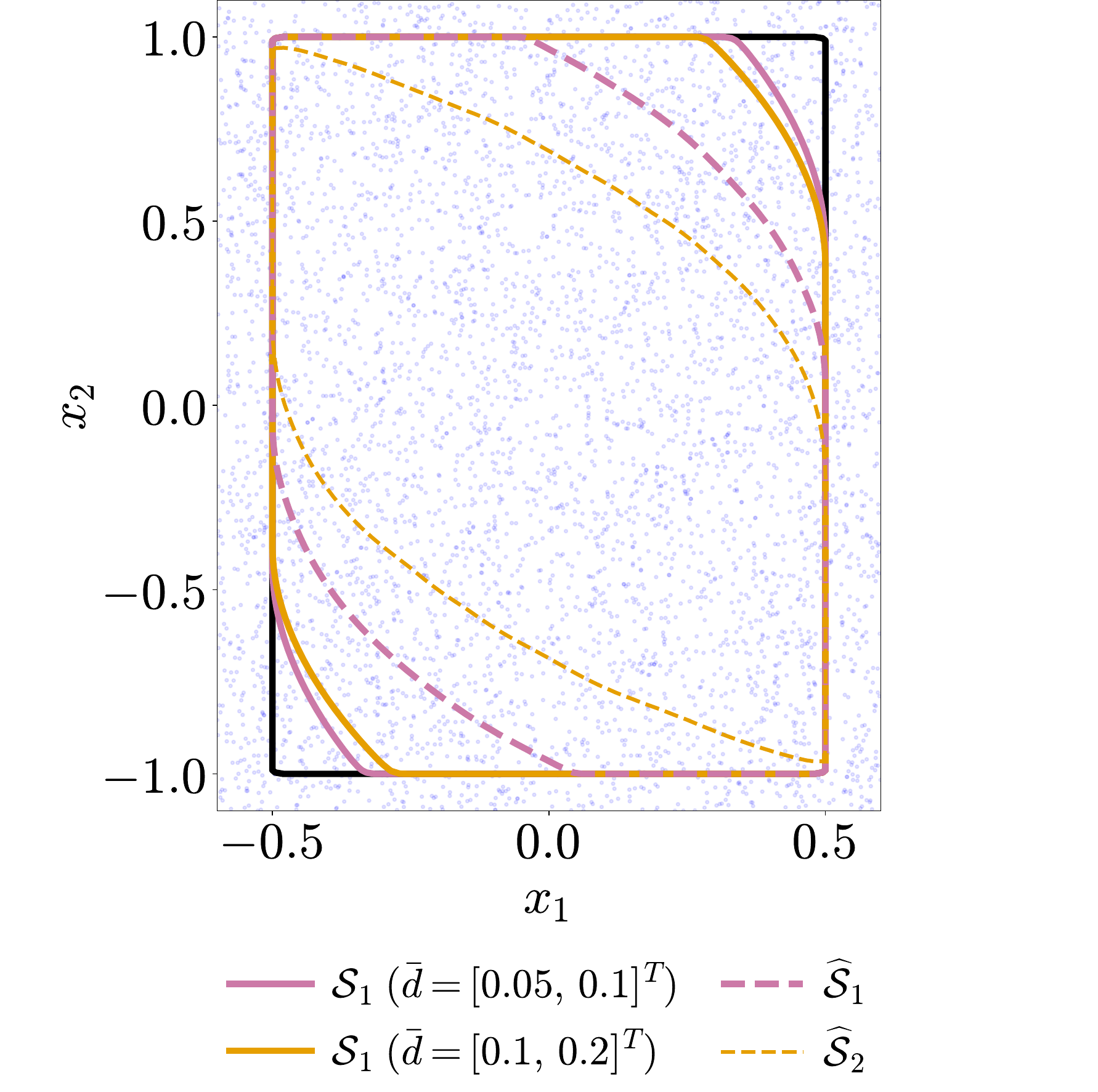}
        \caption{Effect of the disturbance bound \(\bar d\). (Section~\ref{sec:Disturbance Bound})}
        \label{fig:di_disturbance}
    \end{subfigure}
    \hfill
    \caption{Example 1: Double-integrator system under different settings. 
    (a, b) The constraint set, the model-based maximal safe set (Avoid BRT), and data points are illustrated by 
    \includegraphics[
        width=0.2\linewidth,
        trim=0 20pt 0 15pt,
        clip
    ]{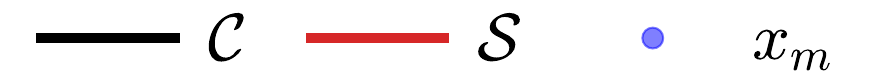}, respectively. 
    (c) Two model-based ground truth Avoid BRTs and their corresponding data-driven counterparts under different disturbance bounds are shown.}
    \label{fig:four-panel}
\end{figure*}

\section{Simulation Examples} \label{sec:results}

This section presents two case studies. The first is a constrained double integrator, which is used to assess the conservatism of the proposed method. The second is an aircraft taxiing model, which is used to certify closed-loop safety under perceptual uncertainty. The HJ time horizon \(t\) in the safety definitions is set to \(2\,\mathrm{s}\) for the double-integrator example and \(40\,\mathrm{s}\) for the aircraft taxiing example. In both cases, initial conditions are sampled uniformly over the prescribed domain, and state velocities are estimated from consecutive state measurements. For the double integrator, each trajectory segment spans one sampling interval \(T_{\mathrm{s}}\), so that \(N_s=2\). For the aircraft model, longer trajectory segments are used, with \(N_s=30\). In the double-integrator study, the control input \(u\) and disturbance \(d\) are randomly selected from the extreme points of \(\mathcal{U}\) and \(\mathcal{D}\), respectively, at the beginning of each segment and are held constant over the segment. In the aircraft taxiing study, the control input is generated by a nonlinear state-feedback law. \rev{Both examples use component-wise Lipschitz constants, except for the double-integrator studies in Subsections~\ref{sec:Sampling Period} and~\ref{sec:Disturbance Bound}, where the uniform Lipschitz representation is used.}

\subsection{Linear Case: Constrained Double Integrator}

Consider the constrained double integrator
\begin{equation} \label{eq:double_integrator}
\begin{aligned}
\dot{\mathbf x}(s) &=
\begin{bmatrix}
0 & 1\\
0 & 0
\end{bmatrix} \mathbf x(s) +\mathbf u(s) + \mathbf d(s), \\[-2pt]
\mathbf x_m(s) &= \mathbf x(s) + \mathbf n_x(s),
\end{aligned}
\end{equation}
where \(\mathbf x(s),\, \mathbf x_m(s),\, \mathbf u(s),\,\mathbf  d(s) \in \mathbb{R}^2\). Control inputs, disturbances, and measurement noise are bounded componentwise:
\begin{equation} \label{eq:constraints_example1}
|\mathbf u(s)| \le \bar{u}, \,
|\mathbf d(s)| \le \bar{d}, \,
|\mathbf n_{x}(s)| \le \sigma_x.
\end{equation}
with \(\bar u =[0.5,1]^\top\). For various \(N_T\), \(\bar d\), \(\sigma_x\), \(L^x \in \mathbb{R}^{2\times2}\), and \(L^d \in \mathbb{R}^{2\times2}\), we compute the Avoid BRT over the time horizon $s \in [-2,0]$. The prescribed domain\footnote{The prescribed domain, \(\mathcal{X}_{\text{d}}\), denotes the computational state-space region over which the Hamilton--Jacobi PDE is solved.}, the constraint and failure sets are defined as follows, respectively.
\[
\mathcal{X}_{\text{d}} = \mathcal{R}\!\left(
\begin{bmatrix}1 \\ 1.5 \end{bmatrix}
\right),\,
\mathcal{X}_{\text{c}} = \mathcal{R}\!\left(
\begin{bmatrix}0.5 \\ 1\end{bmatrix}
\right), \, \mathcal{X}_{\text{f}} = \mathcal{X}_{\text{d}} \setminus 
\mathcal{X}_{\text{c}}.
\]

\subsubsection{Number of Samples (Fig.~\ref{fig:di_N})} \label{Number of Samples}

Let \(\bar d = [0.1,\,0.1]^\top\), \(\sigma_x = [0,\,0]^\top\), and the Lipschitz constants be
\[
L^x_c =
\begin{bmatrix}
0 & 1 \\
0 & 0
\end{bmatrix},
\quad
L^d_c =
\begin{bmatrix}
1 & 0 \\
0 & 1
\end{bmatrix}.
\]
We consider the uncertainty set~\eqref{eq:R-set2}, constructed using the component-wise Lipschitz constants and a small sampling period \(T_{\mathrm{s}} = 0.01\,\mathrm{s}\). 
For \(N_T = 500\), the approximate safe set coincides with the true safe set. 
\rr{This shows that, as the dataset becomes denser, the approximated safe set converges to the true safe set. Notably, this convergence is observed even in the presence of exogenous disturbances, see Theorem \ref{thm:ddh}.}

\subsubsection{Sampling Period (Fig.~\ref{fig:di_Ts})} \label{sec:Sampling Period}
To assess discretization effects, let the sampling period be \(T_{\text{s}} \in \{\,0.05,\,0.36,\,0.8\}\,\mathrm{s}\) and \(N_T = 5000\). 
We use the uniform Lipschitz representation with \(L^x = 1\), \(L^d = 1.41\), and \(M = 2.35\). 
In the presence of noise \(\sigma_x = [0.02,\,0.02]^\top\) and no disturbance \(d = [0,\,0]^\top\), \(T_{\text{s}}^\star = 0.36\,\mathrm{s}\) yields the least conservative safe set. As \(T_{\text{s}}\) increases or decreases, the Avoid BRT correspondingly shrinks, see Section \ref{sec: practical}.

\subsubsection{Disturbance Bound (Fig.~\ref{fig:di_disturbance})} \label{sec:Disturbance Bound}
We consider \(\sigma_x = [0.02,\,0.02]^\top\) and sampling period \(T_{\text{s}}^\star = 0.36\, s\). As expected, larger disturbance bounds \(\bar d_i\) enlarge the uncertainty sets and thereby reduce the size of the safe set.

\subsection{Nonlinear Case: Aircraft Taxiing with a Nonlinear Closed-Loop Controller}

We consider a modified nonlinear model of aircraft taxiing, adopted from \cite{lin2025robust}. The state vector is \(x = [\,p\;\; \theta\,]^\top \in \mathbb{R}^2\), where \(p\) is the lateral position $[\mathrm{m}]$ and \(\theta\) is the heading angle $[\mathrm{rad}]$. The longitudinal speed is constant, \(v_0 = 5\,[\mathrm{\frac{m}{s}}]\), and the control input is the heading rate \(\omega \in \mathbb{R}\), generated by a nonlinear feedback law. The dynamics are
\begin{equation}
\label{eq:nonlinear_dyn}
\dot{\mathbf x}(s) = 
\begin{bmatrix}
v_0 \sin{\boldsymbol{\theta} (s)}\\[2pt]
\boldsymbol{\omega}  (s)
\end{bmatrix},
\end{equation}
with feedback
\begin{equation}
\label{eq:omega_law}
\omega(\hat {\mathbf {p}} (s), \hat {\boldsymbol{\theta}} (s)) = \tan(a\,\hat {\mathbf {p}} (s) + b\, \hat {\boldsymbol{\theta}} (s)),
\end{equation}
where \(a=-0.013\) and \(b=-0.44\) are controller parameters, and \(\hat {\mathbf {p}} (s), \hat {\boldsymbol{\theta}} (s)\) are estimated states defined as follows:
\begin{equation}
    \hat {\mathbf {p}} (s) =  {\mathbf {p}} (s) + \mathbf e_{p}, \quad\hat {\boldsymbol{\theta}} (s) = {\boldsymbol{\theta}} (s) + \mathbf  e_{\theta} (s).
\end{equation}
where \(\mathbf e_{p}(s)\) and \(\mathbf e_{\theta}(s)\) are perceptual uncertainties. We also define the prescribed domain, constraint, failure and target sets as follows:
\[
\begin{aligned}
\mathcal{X}_{\mathrm{d}} 
&= \mathcal{R}\!\left(
\begin{bmatrix}
11\\
0.49
\end{bmatrix}
\right), 
\quad
\mathcal{X}_{\mathrm{c}} 
= \mathcal{R}\!\left(
\begin{bmatrix}
10\\
\infty
\end{bmatrix}
\right),\\
\mathcal{X}_{\mathrm{f}} 
&= \mathcal{X}_{\mathrm{d}} \setminus \mathcal{X}_{\mathrm{c}},
\quad
\mathcal{X}_{\mathcal{T}} 
= \mathcal{R}\!\left(
\begin{bmatrix}
10\\
0.1
\end{bmatrix}
\right).
\end{aligned}
\]

\subsubsection{Model-based solution using HJ}
In the model-based setting, the perceptual uncertainties can be modeled as an exogenous disturbance. Therefore, the closed-loop law \eqref{eq:omega_law} can be re-represented as:
\begin{equation}
\label{eq:omega_law_model_based}
\omega(z, d) = \tan(z + d),
\end{equation}
where \(z = a p + b \theta\) and \(d = a e_p + b e_\theta\). Let the perceptual uncertainties be bounded by \(|e_p|\le 0.2 \ \text{[m]} \) and \(|e_\theta|\le 0.034 \ \text{[rad] (or } 2 \ \text{[deg]})\). Then \(|d| \leq |a|\,|e_p| + |b|\,|e_\theta|
\approx 0.018\).


\subsubsection{Data-driven solution using DDH}
The DDH framework requires upper bounds on the true Lipschitz constants of the system. The Jacobian of~\eqref{eq:nonlinear_dyn} with respect to the state and disturbance is
\begin{equation}
\label{eq:jacobians}
\frac{\partial f}{\partial x} =
\begin{bmatrix}
0 & v_0 \cos\theta\\[2pt]
a\sec^2(z+d) & b\sec^2(z+d)
\end{bmatrix},
\frac{\partial f}{\partial d} =
\begin{bmatrix}
0\\[2pt]
\sec^2(z+d)
\end{bmatrix}.
\end{equation}
Thus, the component-wise Lipschitz bounds are
\begin{align}
\label{eq:jac_bounds_numeric}
\left|\frac{\partial f}{\partial x}\right|
\;\le\;
\begin{bmatrix}
0 & 5\\[2pt]
0.015 & 0.508
\end{bmatrix},
\quad
\left|\frac{\partial f}{\partial d}\right|
\;\le\;
\begin{bmatrix}
0\\[2pt]
1.15
\end{bmatrix}.
\end{align}
Therefore, any choice of \(\big| \tfrac{\partial f}{\partial x} \big| \le L^{x}_\mathrm{c}\) and \(\big| \tfrac{\partial f}{\partial d} \big| \le L^{d}_\mathrm{c}\) is a valid selection for the Lipschitz constants. We sample \(N_T = 2000\) trajectories with initial conditions uniformly from the domain set, $\mathcal{R}\!(
[11,0.49]^\top)$, see Fig. \ref{fig:trajectories}. For each sample, we generate a trajectory using fourth-order Runge–Kutta integration over a fixed time horizon \(30\,T_{\text{s}}\) with sampling period \(T_{\text{s}}^\star=0.26\,\mathrm{s}\). The heading control is computed from \eqref{eq:omega_law}, with \(e_p\) and \(e_\theta\) drawn uniformly within their bounds. Note that in the data-driven setting, there is no need to account for the equivalent disturbance in \eqref{eq:omega_law_model_based} as the effect of perceptual uncertainty is directly considered. The corresponding Reach–Avoid BRT computed using DDH and its exact counterpart are calculated and shown in Fig.~\ref{f2}. \rr{As expected, the robust DDH yields an under-approximation of the true safe set. The observed gap between the DDH set and the ground truth, particularly in the upper-left and lower-right regions, is mainly due to the worst-case treatment of perception error. This assumption is adopted to obtain a deterministic safety certificate, meaning that every admissible perception error within the prescribed uncertainty bound is treated as potentially realizable. Consequently, the resulting safe set is conservative but remains certifiably safe. Reducing this conservatism by replacing the deterministic worst-case uncertainty set with probabilistic error models can significantly reduce conservatism, but at the cost of replacing deterministic safety guarantees with probabilistic ones.}

\begin{figure}[tb]
    \centering
    \includegraphics[width=1\linewidth]{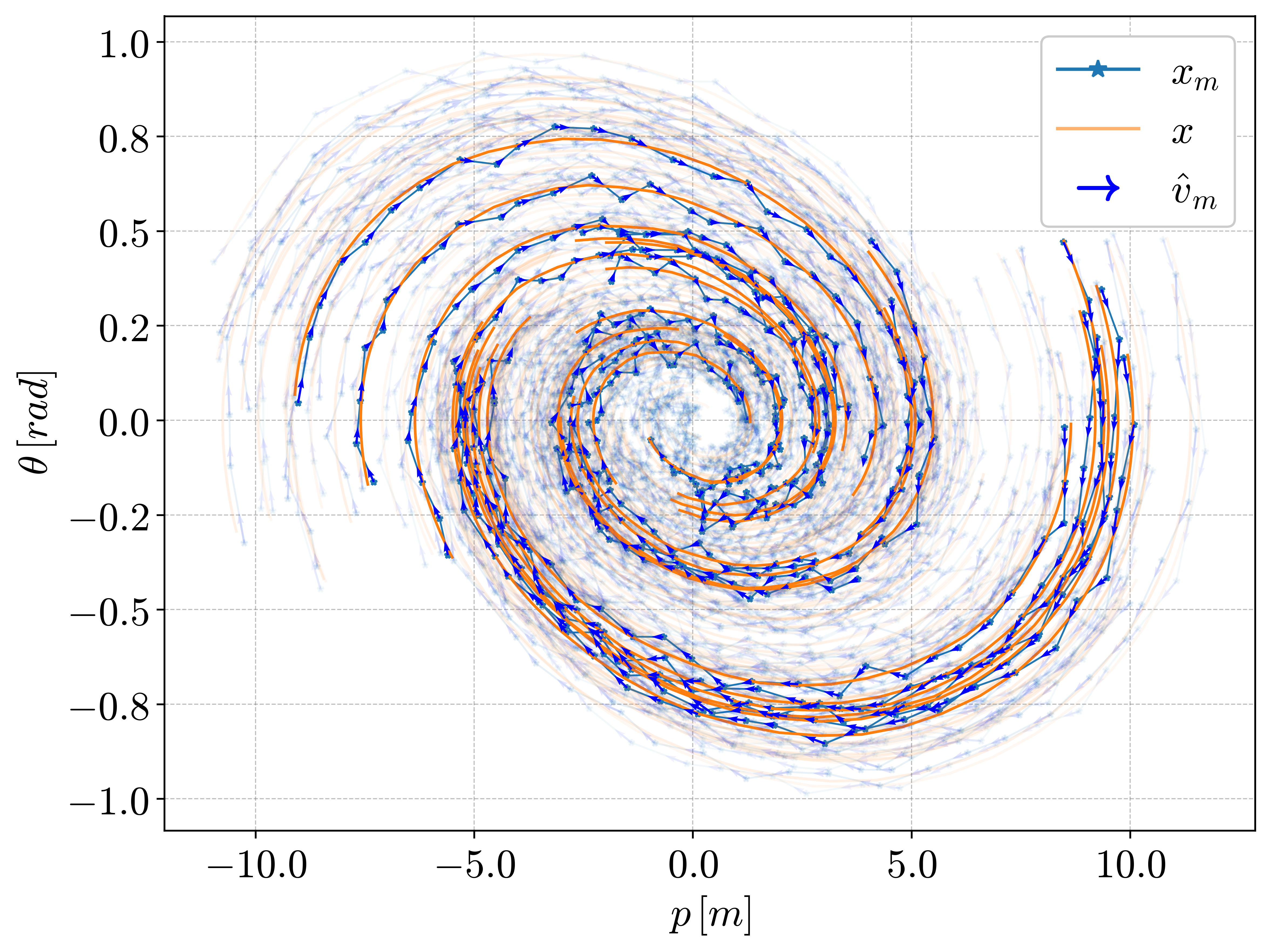}
    \caption{Sampled trajectories of the nonlinear aircraft taxiing system. The points \(x_m\) indicate noisy measurements, while \(\hat{v}_m\) denotes the estimated velocities computed from successive samples. The curve \(x(t)\) represents the underlying continuous-time trajectory.}
    \label{fig:trajectories}
\end{figure}
\begin{figure}[tb]
    \centering
    \includegraphics[width=1\linewidth]{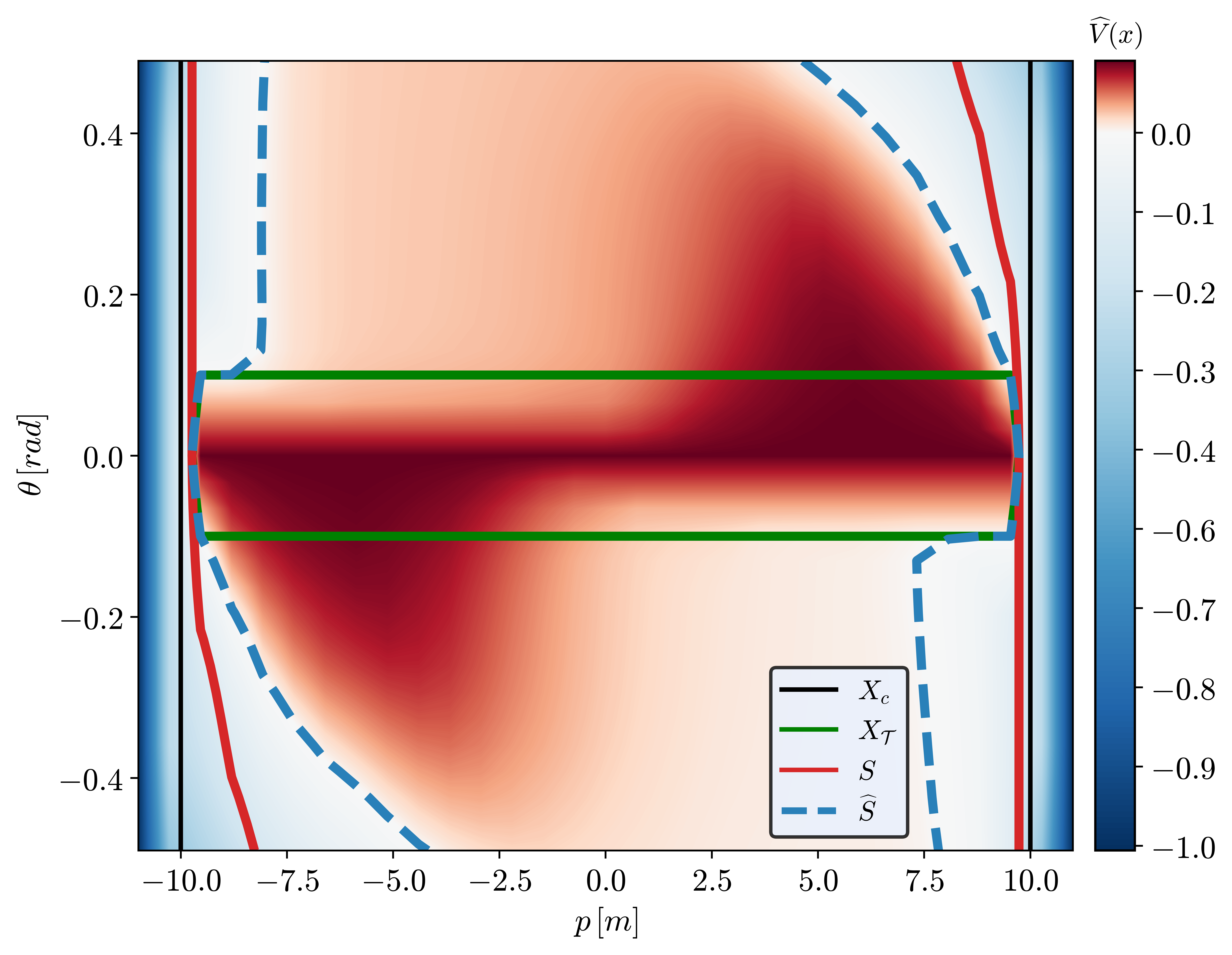}
\caption{Reach-Avoid BRT for the aircraft taxiing system. Comparison between the DDH-based approximation and the exact solution. Sets \(X_c\), \(X_{\mathcal{T}}\), \(S\), and \(\widehat{S}\) represent the constraint set, the target set, the exact reach-avoid safe set, and the data-driven reach-avoid safe set, respectively.}
    \label{f2}
\end{figure}

\section{Conclusion}
\label{sec:conclusion}

This paper presented a Robust Data-Driven Hamiltonian (R-DDH) reachability framework for safe set computation directly from noisy measurements. By constructing a conservative Hamiltonian via Lipschitz-based uncertainty sets, we derived a rigorous lower bound on the exact Hamiltonian, enabling inner approximations of the maximal robust safe set. The proposed formulation explicitly accounts for unknown but bounded exogenous disturbances, measurement noise, and sampling-induced effects, thereby robustifying the framework in \cite{choi2025data} and extending its practical applicability. Additionally, it provides practical guidelines for selecting an appropriate sampling period.

To reduce the conservatism inherent in the worst-case uncertainty assumption, future research will extend the DDH framework to incorporate probabilistic uncertainty representations. In addition, a possible direction is to develop data-efficient experimental design strategies based on active safe exploration policies, with the goal of reducing redundancy in data collection. This is motivated by the empirical observation that the Hamiltonian is typically determined by a small subset of extreme data points. These enhancements will facilitate the scalability of the R-DDH framework to higher-dimensional systems. Finally, these advances make R-DDH applicable to a broader range of safety-critical real-world control systems.

\balance

\bibliographystyle{IEEEtran}
\bibliography{Ref}

\end{document}